\documentclass[preprint,12pt]{elsarticle}
\usepackage{amsmath,amsthm}
\usepackage{latexsym}
\usepackage{amssymb}
\usepackage[table]{xcolor}
\usepackage{graphicx}
\usepackage{verbatim}
\usepackage{xspace}
\usepackage{subfigure} 
\usepackage{algorithm}
\usepackage[noend]{algorithmic}
\usepackage[tableposition=top]{caption}
\usepackage{tikz}
\usepackage{url}
\usetikzlibrary{positioning,chains,fit,shapes,calc,arrows}

\newcommand{\pspace}{\textsc{PSPACE}\xspace}
\newcommand{\np}{\textsc{NP}\xspace}
\newcommand{\IS}{\textsc{Online Independent Set}\xspace}
\newcommand{\MOS}{\textsc{Maximum Online Set}\xspace}
\newcommand{\ISoff}{\textsc{Independent Set}\xspace}
\newcommand{\VC}{\textsc{Online Vertex Cover}\xspace}
\newcommand{\DS}{\textsc{Online Dominating Set}\xspace}
\newcommand{\M}{\textsc{Online Matching}\xspace}
\newcommand{\IM}{\textsc{Online Matroid Intersection}\xspace}
\newcommand{\F}{\textsc{Online Maximal Forest}\xspace}

\def\GI{\ensuremath{\texttt{GIS}}\xspace}
\def\GMOS{\ensuremath{\texttt{GMOS}}\xspace}

\def\AGI{\ensuremath{\texttt{Almost-GIS}}\xspace}
\def\GV{\ensuremath{\texttt{GVC}}\xspace}
\def\GD{\ensuremath{\texttt{GDS}}\xspace}
\def\GM{\ensuremath{\texttt{GM}}\xspace}

\def\GF{\ensuremath{\texttt{GF}}\xspace}
\def\optalg{\ensuremath{\texttt{IS-STAR}}\xspace}

\def\alg{\ensuremath{\texttt{ALG}}\xspace}
\def\adv{\ensuremath{\texttt{ADV}}\xspace}
\def\opt{\ensuremath{\texttt{OPT}}\xspace}
\def\opton{\ensuremath{\texttt{OPTON}}\xspace}
\def\algp{\ensuremath{\overline{\texttt{ALG}}}\xspace}
\def\optalgp{\ensuremath{\overline{\texttt{IS-STAR}}}\xspace}

\def\ois{\ensuremath{I^O}\xspace}
\def\ovc{\ensuremath{V^O}\xspace}
\def\ods{\ensuremath{D^O}\xspace}
\def\om{\ensuremath{M^O}\xspace}

\def\oms{\ensuremath{\mathit{MS}^O}\xspace}
\newcommand{\lille}{\textsc{$s(G')$}\xspace}
\newcommand{\family}{\textsc{$\tilde{G}$}\xspace}
\newcommand{\stor}{\textsc{$b(G')$}\xspace}

\usepackage{array}
\newtheorem{theorem}{Theorem}[section]
\newtheorem{lemma}[theorem]{Lemma}

\newtheorem{corollary}[theorem]{Corollary}
\newtheorem{definition}[theorem]{Definition}

\journal{Discrete Applied Mathematics}

\begin{document}
\newtheorem{assumption}{Assumption}
\newtheorem{observation}[theorem]{Observation}
\begin{frontmatter}
\title{Adding Isolated Vertices Makes some Greedy Online Algorithms Optimal\tnoteref{t1}}
\author[sdu]{Joan Boyar\corref{c1}}
\ead{joan@imada.sdu.dk}
\author[sdu]{Christian Kudahl}
\ead{kudahl@gmail.com}
\tnotetext[t1]{A preliminary version of this paper appeared in the proceedings of the
26th International Workshop on Combinatorial Algorithms (IWOCA 2015),
{\em Lecture Notes in Computer Science}, 9538: 65-76, Springer 2016.}
\cortext[c1]{Corresponding author}
\address[sdu]{Department of Mathematics and Computer Science,
University of Southern Denmark, Campusvej 55, 5270 Odense M, Denmark}

\begin{abstract}
An unexpected difference between online and offline algorithms
is observed.
The natural greedy algorithms are shown to be worst case online optimal 
for \IS and \VC on graphs with ``enough'' isolated vertices, Freckle
Graphs.
For \DS, the greedy algorithm is shown to be worst case online optimal
on graphs with at least one isolated vertex. 
These algorithms are not online optimal in general. 
The online optimality results for these greedy algorithms imply optimality
according to various worst case performance measures, such as the
competitive ratio. It is also shown that, despite this worst case
optimality, there are Freckle graphs where the greedy independent set 
algorithm is objectively less good than another algorithm.

It is shown that it is \np-hard to determine any of the following
for a given graph:
the online independence number, the online vertex cover number, and the 
online domination number.
\end{abstract}

\end{frontmatter}
\section{Introduction}
This paper contributes to the larger goal of better understanding the nature of online optimality,
greedy algorithms, and different performance measures for online algorithms.
The graph problems \IS, \VC and \DS, which are defined below, are
considered in the \emph{vertex-arrival} model, where
the vertices of a graph, $G$, are revealed one by one.
When a vertex is revealed (we also say that it is ``requested''), its edges to previously revealed vertices
are revealed.
At this point, an algorithm irrevocably either accepts the vertex or rejects it.
This model is well-studied (see for example,
\cite{LST89,GL90,Vishwanathan1992657,GKL97,conjecture,Halldorsson99onlinecoloring,Halldorsson2002953}).

We show that, for some graphs, an obvious greedy algorithm for each
of these 
problems performs less well than another online algorithm and thus
is not online optimal. However, this greedy algorithm performs (at least in
some sense) at least
as well as any other online algorithm for these problems, as long
as the graph has enough isolated vertices. Thus, in contrast to
the case with offline algorithms, adding isolated vertices to
a graph can improve an algorithm's performance, even making it
``optimal''.

For an online algorithm for these problems and a particular
sequence of requests, let $S$ denote the set of 
accepted vertices, which we call a \emph{solution}.
When all vertices have been revealed (requested and either accepted
or rejected by the algorithm), $S$ must fulfill certain conditions:

\begin{itemize}
\item In the \IS problem \cite{Halldorsson2002953,DemangePP00}, $S$ must form an independent set. That is, no two
vertices in $S$ may have an edge between them. The goal is to maximize $|S|$.
\item In the \VC problem \cite{Demange200583}, $S$ must form a vertex cover. That is, each edge
in $G$ must have at least one endpoint in $S$. The goal is to minimize $|S|$.
\item In the \DS problem \cite{dominatingset}, $S$ must form a dominating set. That is, each vertex
in $G$ must be in $S$ or have a neighbor in $S$. The goal is to minimize $|S|$.
\end{itemize}

If a solution does not live up to the specified requirement, it is said to be infeasible.
The score of a feasible solution is $|S|$. The score of an infeasible solution is $\infty$
for minimization problems and $-\infty$ for maximization problems. Note that for \DS, it is not required that $S$ form a dominating set at all times. It just needs to be a dominating set
when the whole graph has been revealed.
If, for example, it is known that the graph is connected, the algorithm might reject the first vertex since it is known that it will be possible to dominate this vertex later.

In Section~\ref{prelims}, we define the greedy algorithms for
the above problems, along with concepts analogous to the
online chromatic number of Gy\'{a}rf\'{a}s et al.~\cite{GKL99}
for the above problems, giving a natural definition of optimality for
online algorithms. In Section~\ref{nonopt}, we show that
greedy algorithms are not in general online optimal for these problems.
In Section~\ref{freckle}, we define Freckle Graphs, which are
graphs which have ``enough'' isolated vertices to make the 
greedy algorithms online optimal. In proving that the greedy algorithms
are optimal on Freckle Graphs, we also show that, for \IS, one can,
without loss of generality, only consider adversaries which never request
a vertex adjacent to an already accepted vertex, while there are alternatives.
In Section~\ref{adding}, we investigate what other online problems
have the property that adding isolated requests make greedy algorithms optimal. 
In Section~\ref{measures}, it is shown that
the online optimality results for these greedy algorithms imply optimality
according to various worst case performance measures, such as the
competitive ratio. In Section~\ref{bijective}, it is shown that, despite this worst case
optimality, there is a family of Freckle graphs where the greedy 
independent set 
algorithm is objectively less good than another algorithm.
Various NP-hardness
results concerning optimality are proven in Section~\ref{hardness}.
There are some concluding remarks and open questions in the last section.
Note that Theorem~\ref{ishard} and Theorem~\ref{ispspace} appeared in the second author's Master's
thesis~\cite{kudahl}, which served as inspiration for this paper.

\section{Algorithms and Preliminaries}
\label{prelims}
For each of the three problems, we define a greedy algorithm.
\begin{itemize}
\item In \IS, \GI accepts a revealed vertex, $v$, iff no neighbors of $v$ have been accepted.
\item In \VC, \GV accepts a revealed vertex, $v$, iff a neighbor of $v$ has previously been revealed 
but  not accepted.
\item In \DS, \GD accepts a revealed vertex, $v$, iff no neighbors of $v$ have been accepted.
\end{itemize}

Note that the algorithms \GI and \GD are the same (they have different names to emphasize that they solve different problems).
For an algorithm \alg, we define \algp to be the algorithm that simulates \alg and accepts exactly those vertices
that \alg rejects. This defines a bijection between \IS and \VC algorithms.
Note that $\GV = \overline{\GI}$.

For a graph, $G$, an ordering of the vertices, $\phi$, and an algorithm, \alg, we let $\alg(\phi(G))$ denote the
score of \alg on $G$ when the vertices are requested in the order $\phi$. We let $|G|$ denote the number of vertices in $G$.

For minimization problems, we define:
\[
	\alg(G)= \max_\phi \alg(\phi(G))
\]
That is, $\alg(G)$ is the highest score $\alg$ can achieve over all orderings of the vertices in $G$.

For maximization problems, we define:
\[
	\alg(G)= \min_\phi \alg(\phi(G))
\]
That is, $\alg(G)$ is the lowest score $\alg$ can achieve over all orderings of the vertices in $G$.

Since we consider a worst possible ordering, we sometimes think of an adversary
as  ordering the vertices.

\begin{observation} \label{isvc}
Let \alg be an algorithm for \IS. Let a graph, $G$, with $n$ vertices be given.
Now, \algp is an \VC algorithm and $\alg(G)+\algp(G)=n$.
\end{observation}
The equality $\alg(G)+\algp(G)=n$ holds, since a worst ordering of $G$ for \alg is also a worst ordering for $\algp$.

In considering online algorithms for coloring, \cite{GKL99} defines
the online chromatic number, which intuitively is the best result
(minimum number of colors) any online algorithm can be guaranteed
to obtain for a particular graph (even when the graph, but not the ordering, is known in advance). We define analogous concepts
for the problems we consider, defining for every graph a number representing the best value any online algorithm can achieve.
Note that in considering all algorithms, we include those which know 
the graph in advance.
Of course, when the graph is known, the order in which the vertices are requested is not known to an online algorithm, and
the  label given with a requested vertex does not necessarily correspond to its
label in the known graph:
The subgraph revealed up to this point might be isomorphic
to more than one subgraph of the known graph and it could correspond
to any of these subgraphs. 

Let $\ois(G)$ denote the \emph{online independence number} of $G$. This is the largest number such that
there exists an algorithm, \alg, for \IS with $\alg(G)=\ois(G)$.
Similarly,  let $\ovc(G)$, the \emph{online vertex cover number}, be the smallest number such that there exists an algorithm, \alg, for \VC with $\alg(G)=\ovc(G)$.
Also let $\ods(G)$, the \emph{online domination number}, be the smallest number such that there exists an algorithm, \alg, for \DS with $\alg(G)=\ods(G)$.

The same relation between the online independence number and the
online vertex cover number holds as between the independence number and
the vertex cover number.

\begin{observation} \label{isvcsum} 
For a graph, $G$ with $n$ vertices, we have $\ois(G)+\ovc(G)=n$.
\end{observation}
\begin{proof}
Let a graph, $G$, with $n$ vertices be given.
Let $\alg$ be an algorithm for \IS such that $\alg(G)=\ois(G)$. 
From Observation \ref{isvc}, we have that $\algp$ is an algorithm for \VC such 
that $\algp(G)=n-\ois(G)$. It must hold that $\algp(G)=\ovc(G)$, since the 
existence
of an algorithm with a lower vertex cover number
would imply the existence of a corresponding algorithm for \IS with an 
independence number greater than $\alg(G)=\ois(G)$. \qed
\end{proof}

\section{Non-optimality of Greedy Algorithms}
\label{nonopt}
We start by motivating the other results in this paper by showing 
that the greedy algorithms are not optimal in general. In particular,
they are not optimal on the star graphs, $S_n$, $n\geq 3$, which
have a center vertex, $s$, and $n$ other vertices, adjacent to
$s$, but not to each other.

The algorithm, \optalg (see Algorithm~\ref{IS-STAR}), does much better than \GI for the
independent set problem on star graphs.

\begin{algorithm}[!t] 
\begin{algorithmic}[1]
\FOR{request to vertex $v$}
\IF{$v$ is the first vertex}
\STATE reject $v$
\ELSIF{$v$ is the second vertex and it has an edge to the first}
\STATE reject $v$
\ELSIF{$v$ has more than one neighbor already}
\STATE reject $v$
\ELSE 
\STATE accept $v$
\ENDIF
\ENDFOR
\end{algorithmic}
\caption{\optalg, an online optimal algorithm for independent set for $S_n$}
\label{IS-STAR}
\end{algorithm}

\begin{theorem} \label{optvsgreedy}
For a star graph, $S_n$,   $\optalg(S_n)=n-1$ and $\GI(S_n) =1$.
\end{theorem}
\begin{proof}
We first show that \optalg never  accepts the center vertex, $s$.
If $s$ is presented first, it will be rejected.
If it is presented second, it will have an edge to the first vertex and be rejected.
If it is presented later, it will have more than one neighbor and be rejected.
Since $\optalg$ never accepts $s$, it produces an independent set.
For every ordering of the vertices, $\optalg$ will reject the
first vertex. If the first vertex is $s$, it will reject the
second vertex. Otherwise, it will reject $s$ when it comes. Thus,
$\optalg(S_n)=n-1$.
On the other hand, $\GI(G)=1$, since it will accept $s$ if it is 
requested first. \qed
\end{proof}

Since $n-1 >1$ for $n\geq 3$, we can conclude that \GI is not
an optimal online algorithm for all graph classes.
\begin{corollary} \label{isnotopt}
For \IS, there exists an infinite family of graphs, $S_n$ for $n\geq 3$, 
and an online algorithm, \optalg, such that $\GI(S_n)$ $<$ $\optalg(S_n)$.
\end{corollary}
Note that if some algorithm, \alg, rejects the first vertex requested,
$\alg(S_n) \leq n-1$, and if it accepts the first vertex, $\alg(S_n)=1$.
Thus \optalg is optimal.

To show that \GV is not an optimal algorithm for \VC, we consider 
\optalgp.
\begin{corollary} \label{vcnotopt}
For \VC, there exists an infinite family of graphs, $S_n$ for $n\geq 3$, 
and an online algorithm, \optalgp, such that $\optalgp(S_n)$ $<$ $\GV(S_n)$.
\end{corollary}
\begin{proof}
Using Observation~\ref{isvc} and Theorem~\ref{optvsgreedy}, 
we have that $\optalgp(S_n)=n+1-\optalg(S_n)=2$ and
$\GV(S_n)=n+1-\GI(S_n)=n$. \qed
\end{proof}

Finally, for \DS, we have a similar result.

\begin{corollary} \label{dsotopt}
For \DS, there exists an infinite family of graphs, $S_n$ for $n\geq 3$, 
and an online algorithm, \optalgp, such that $\optalgp(S_n)$ $<$ $\GD(S_n)$.
\end{corollary}
\begin{proof}
Requesting $s$ last
ensures that $\GD$ accepts $n$ vertices. It can never accept
all $n+1$ vertices, so $\GD(S_n)=n$.
On the other hand, $\optalgp(S_n)=2$ (as in the proof of 
Corollary~\ref{vcnotopt}).
We note that a vertex cover is also a dominating set in connected graphs. 
This means that $\optalgp$ always produces a dominating set in $S_n$. \qed
\end{proof}

\section{Optimality of Greedy Algorithms on Freckle Graphs}
\label{freckle}
For a graph, $G$, we let
\begin{itemize}
\item $k$ denote the number of isolated vertices,
\item $G'$ denote the graph induced by the non-isolated vertices,
\item \stor be a maximum independent set in $G'$, and
\item \lille be a minimum inclusion-maximal independent set in $G'$ 
(that is, a smallest independent set such that including any additional vertex in the set would cause it to no longer be independent). 
\end{itemize}
Note that $|\lille|$ is also known as the \emph{independent domination number} of $G'$ (see \cite{idn} for more information).

Using this notation, we define the following class of graphs.

\begin{definition}
A graph, $G$, is a
\emph{Freckle Graph} if $k+|\lille| \geq \ois(G')$.
\end{definition}
Note that all graphs where at least half the vertices are isolated are
Freckle Graphs.
If the definition was changed to this (which might be less artificial), the results presented here would still hold, but our definition gives stronger results.
The name comes from the idea that such a graph in many cases has a lot of isolated 
vertices (freckles). Furthermore, any graph can be turned into a Freckle Graph by adding enough isolated vertices.
Note that a complete graph is a Freckle Graph. To make the star graph, $S_n$, a freckle graph, we
need to add $n-2$ isolated vertices.
We show that \GI and \GV are online optimal on all Freckle Graphs.
For the proof, we need a little more terminology and a helpful lemma.

\begin{definition}
A request is \emph{pointless} if it is to a vertex which has a neighbor which was already accepted.
\end{definition}

\begin{definition}
For a graph, $G$, an adversary is said to be \emph{conservative} if it does not make pointless requests unless only such requests remain.
\end{definition}

\begin{lemma}\label{conservative}
For \IS,
for every graph, $G$, there exists a conservative adversary, \adv, which
ensures that every algorithm accepts an independent set in $G$ of size at most $\ois(G)$.
\end{lemma}
\begin{proof}
Assume, for the sake of contradiction, that there exists an algorithm \alg,
which accepts an independent set of size at least $\ois(G)+1$ against
every conservative adversary.
We now describe an algorithm, $\alg'$, which accepts an independent set of size at least 
$\ois(G)+1$ against any adversary. This contradicts the definition of $\ois(G)$.

Intuitively, since pointless requests must be rejected by any algorithm, 
$\alg'$ can reject pointless requests and otherwise ignore them, reacting
as $\alg$ would against a conservative adversary on the other requests.
$\alg'$ works as follows: It maintains a virtual graph, $G'$, which, 
inductively, is a copy of
the part of $G$ revealed so far, but without the pointless requests.
When a new non-pointless vertex is requested, the same vertex is added to $G'$,
including only the edges to previous vertices which are not pointless
(the pointless requests are not in $G'$).
$\alg'$ now accepts this request if \alg accepts the corresponding request 
in $G'$. When a pointless request is made, $\alg'$ rejects it and does not 
add it to $G'$.

Note that every time \alg accepts a vertex in $G'$, $\alg'$ accepts the 
corresponding vertex in $G$.
Thus, $\alg'(G) \geq \alg(G') \geq \ois(G)+1$ which is a contradiction.
\qed
\end{proof}

\begin{theorem} \label{isgod}
For any algorithm, \alg, for \IS, and for any Freckle Graph, $G$, 
$\GI(G) \geq \alg(G)$.
\end{theorem}
\begin{proof}
First, we note that \GI will accept the $k$ isolated vertices. In $G'$, it will accept an inclusion-maximal independent set.
Since we take the worst ordering, it accepts $|\lille|$ vertices. We get $\GI(G)=k+|\lille|$.
Now we describe an adversary strategy which ensures that an arbitrary algorithm, \alg, accepts at most $k+|\lille|$ vertices.

The adversary starts by presenting 
isolated vertices until \alg
either accepts $|\lille|$ vertices or rejects $k$ vertices. 

If \alg accepts $|\lille|$ vertices, the adversary decides
that they are exactly those in \lille. This means that \alg will accept
no other vertices in $G'$. Thus, it accepts at most $k+\lille$ vertices.

If \alg rejects $k$ vertices, the adversary decides that they are the $k$ isolated vertices.
We now consider $G'$.
At this point, up to $|\lille|-1$ isolated vertices may have been requested and accepted.
Using Lemma~\ref{conservative}, we see that requesting independent vertices
up to this point is optimal play from an adversary playing against an algorithm
which has accepted all of these isolated requests. Following this optimal 
conservative adversary strategy ensures that the algorithm accepts an 
independent set of size at most $\ois(G') \leq k+|\lille| = \GI(G)$.
\qed
\end{proof}

\begin{corollary}
\label{freckle;greedy}
For any Freckle Graph, $G$,  $\GI(G)=\ois(G)$.
\end{corollary}

Intuitively, \GI becomes optimal on Freckle Graphs because
the isolated vertices allow it to accept a larger independent
set, even though it still does poorly on the connected part of the graph.
Any algorithm, which outperforms \GI on the connected part of the graph, must
reject a large number of the isolated vertices in order to keep this advantage.

In contrast, for vertex cover adding isolated vertices to a graph does not make
\GV accept fewer vertices. \GV becomes optimal on Freckle
Graphs because the isolated vertices force any other online algorithm
to accept some of those isolated vertices.

\begin{corollary} \label{vcgod}
For any algorithm, \alg, for \VC, and for any Freckle Graph, $G$, 
$\GV(G) \leq \alg(G)$.
\end{corollary}
\begin{proof}
This follows from Theorem~\ref{isgod}, Observation~\ref{isvc}, and the fact that $\GV = \overline{\GI}$. \qed
\end{proof}

\begin{corollary}
For any Freckle Graph, $G$,  $\GV(G)=\ovc(G)$.
\end{corollary}

For \DS something similar holds, but only one isolated vertex is needed.
\GD becomes optimal because any dominating set has to include that
isolated vertex.

\begin{theorem} \label{dsgod}
For any algorithm, \alg, for \DS and for any graph, $G$, with at least one isolated vertex, 
$\GD(G) \leq \alg(G)$.
\end{theorem}
\begin{proof}
Recall that $k$ denotes the number of isolated vertices in $G$,
and $G'$ denotes the subgraph of $G$ induced by the non-isolated
vertices.
Note that \GD always produces an independent set. Thus, \GD 
accepts at most $k+|\stor|$ vertices; it accepts exactly 
the $k$ isolated vertices and the vertices in \stor
if these are presented first.

Let an algorithm, \alg, be given. The adversary can start by presenting $k+|\stor|$ isolated vertices. If at least one of these vertices is not accepted by \alg, the adversary can
decide that this was in fact an isolated vertex, which can now no longer be dominated. Thus, $\alg(G)= \infty$. If \alg accepts all the presented vertices, it gets
a score of at least $k+|\stor|$. \qed
\end{proof}

\begin{corollary}
For any graph, $G$, with an isolated vertex,  $\GD(G)=\ods(G)$.
\end{corollary}

\section{Adding Isolated Elements in Other Problems}\label{adding}

These results, showing that adding isolated vertices to a graph
can make the greedy algorithms for \IS and \VC optimal, lead one
to ask if similar results hold for other problems. The answer is
clearly ``yes'': We give similar results for \M and \MOS 
(including \IM as a special case).

We consider \M in the edge-arrival model, so each request is an
 edge which must be accepted or rejected.
 If one or both of the vertices that are endpoints of the edge
 have not been revealed yet, they are revealed with the edge.
The goal is to
accept as large a set, $S$, as possible, under the restriction
that $S$ is a matching. Thus, no two edges in $S$ can be incident
to each other.
One can define $\om(G)$, the \emph{online matching
number} of $G$, analogously to the online independence number,
to be the largest number such that there exists an algorithm, \alg,
for \M with $\alg(G)=\om(G)$. Let \GM be the natural greedy algorithm
for \M, which accepts any edge not incident to any edge already
accepted.  Instead of adding isolated vertices, we add
isolated edges, edges which do not share any vertices with any
other edges. 
The number of isolated edges to add would be $k$,
where $\om(G) \leq \GM(G)+k$. We get the following theorem:
Let $G'$ denote the graph $G$ induced by the non-isolated edges.

\begin{theorem}\label{matching}
Let $G$ be a graph where $\om(G') \leq \GM(G')+k$.
For \M, we have that
\[
	\GM(G)=\om(G).
\]
\end{theorem}
\begin{proof}
 Note that a matching in a graph $G=(V,E)$
corresponds to an independent set in the line graph $L(G)$,
where the vertices of $L(G)$ correspond to the edges of $G$,
 and two vertices of $L(G)$ are adjacent, if and only if the
corresponding edges are incident to each other in $G$. Thus,
since \GI is optimal for the graph with $\ois(L(G))-\GI(L(G))$
isolated vertices (or more), \GM is optimal for the graph with $\om(G')-\GM(G')$
isolated edges (or more). \qed
\end{proof}


All of the above problems are in the class AOC~\cite{BFKM15}, so
one is tempted to ask if all problems in AOC have a similar property,
or if all maximization problems in AOC do. This is not the case.
\begin{definition} \label{sgeasydef}\label{wdef}
A problem is in AOC \emph{(Asymmetric Online Covering)} if the following
hold:
\begin{itemize}
\item Each request must be either accepted or rejected on arrival.
\item The cost (profit) of a feasible solution is the number of accepted requests.
\item The cost (profit) of an infeasible solution is $\infty$ ($-\infty$).
\item For any request sequence, there exists at least one feasible solution.
\item A superset (subset) of a minimum cost (maximum profit) 
solution is feasible.
\end{itemize}
\end{definition}

An upper bound on the advice complexity of all problems in AOC was
proven in \cite{BFKM15}, along with a matching lower bound for a
subset of these problems, the AOC-complete problems.
\begin{theorem}
There exists a maximization problem in the class AOC, where adding
isolated  requests which are independent of all others in the sense
that these requests can be added to any feasible set, maintaining feasibility,
does not make the natural greedy algorithm optimal.
\end{theorem}
\begin{proof}
Consider the problem \F, in the vertex arrival model, where the goal
is to accept as large a set, $S$, of vertices, as possible, under the 
restriction that $S$ may not contain a cycle. Consider the following
graph, $G'_n=(V,E)$, where 
\begin{align*}
V=&\{ x,y,v_1,v_2,\ldots,v_n \} \text{ and} \\
E=& \{ (x,y)\} \cup \{ (x,v_i),(y,v_i)\mid 1\leq i\leq n\}.
\end{align*}
Figure~\ref{G5} shows $G'_5$.

\begin{figure}[h]
\centering
\begin{tikzpicture}[-,>=stealth',auto,node distance=1cm,
  thick,
  main node/.style={circle,fill=blue!20,minimum size=0.75cm,draw,font=\sffamily\Large\bfseries},
  uncolored/.style={circle,fill=white!20,minimum size=0.75cm,draw,font=\sffamily\Large\bfseries},
  smallnode/.style={circle,fill=blue!20,minimum size=0.75cm,draw,font=\sffamily\Large\bfseries}]

  \node(1)[main node] {$x$};
  \node(x)[below = 1.5 cm of 1] {};
  \node(2)[main node] [right = 0.1 cm of x] {$v_1$};
  \node(3)[main node] [right = 0.5 cm of 2] {$v_2$};
  \node(4)[main node] [right = 0.6 cm of 3] {$v_3$};
  \node(5)[main node] [right = 1 cm of 4] {$v_4$};
  \node(6)[main node] [right = 1.5 cm of 5] {$v_5$};
  
  \node(7)[main node] [below = 1.5 cm of x] {$y$};

  \path[every node/.style={font=\sffamily\small}]
	(1) edge (2)
	(1) edge (3)
	(1) edge (4)
	(1) edge (5)
	(1) edge (6)
	
	(7) edge (2)
	(7) edge (3)
	(7) edge (4)
	(7) edge (5)
	(7) edge (6)
	;
  \path[every node/.style={font=\sffamily\small}]
	(1) edge (7) [bend left]
	;
\end{tikzpicture}
\caption{The graph $G'_5$.} \label{G5}
\end{figure}
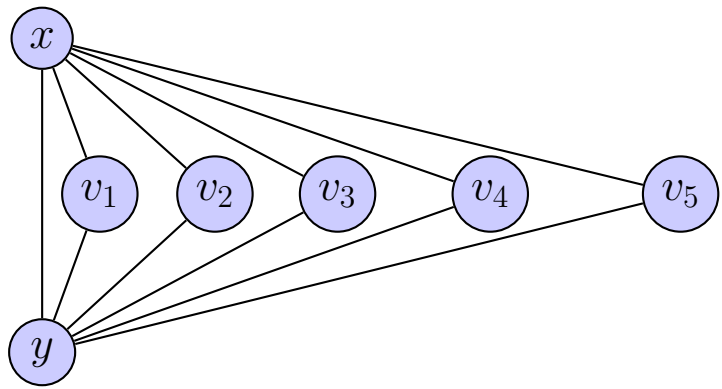

We let $W=\{v_1,v_2, \ldots, v_n\}$.
Consider $G_n$ which is $G'_n$ with an arbitrary number $k$ of isolated
vertices added. Let \GF be the natural greedy algorithm for \F,
which accepts any vertex which does not create a cycle. If the
adversary requests $x$ and $y$ before any vertex in $W$, \GF
cannot accept any vertex in $W$, so $\GF(G'_n)=k+2$. But there
is another algorithm, \alg, which accepts more. 
The algorithm \alg accepts a vertex $v$ if
\begin{itemize}
\item $v$ has degree at most two and
\item all neighbors of $v$ have degree at least three (degree two before the current request).
\end{itemize}
We claim that \alg cannot accept both $x$ and $y$.
Assume $x$ was requested before $y$ and accepted.
Now, $y$ can only be accepted if $x$ already has two other neighbors, $v_i$ and $v_j$,
when $y$ is requested. However, this means that $y$ will also have these two neighbors and
\alg will reject it because its degree is at least three. The argument is symmetric if $y$ is requested before $x$.

We now show that \alg accepts at least $k+n-1$ vertices.
The $k$ isolated vertices will be accepted by \alg regardless of when they are requested.

Now assume $x$ is requested before all vertices in $W$ and before $y$.
In this case, $x$ is accepted. We have shown that $y$ will be rejected when it is requested.
At most two vertices from $W$ can be rejected. When at least two vertices from $W$ have already been requested
and a new vertex $v_i$ is requested, it holds that $v_i$ has degree at most two and that all neighbors of $v_i$ ($x$ and possibly $y$)
have degree at least three. Thus, $v_i$ is accepted. In total, at least $k+1+n-2=k+n-1$ vertices are accepted. 
A symmetric situation holds if $y$ is presented before $x$ and all vertices in $W$.

We now consider the case where a vertex, $v_i \in W$ is requested before $x$ and $y$.
In this case, $v_i$ is accepted. When $x$ and $y$ are requested, they will be rejected since they have a neighbor ($v_i$) whose
degree it at most two. At most one vertex in $W$ can be rejected, since after that, two vertices in $W$ have already been requested 
($v_i$ was requested first). When another vertex $v_j$ is requested, it holds that its possible neighbors ($x$ and $y$) have degree at least three.
In total, at least $k+n-1$ vertices are accepted. 

Hence, the greedy algorithm is not optimal for $G_n$ with
$n \geq 4$. \qed
\end{proof}

We now consider another class of problems where the property
does hold. This is a generalization of \IS.
We consider the problem \MOS.
In this problem, an instance consists of a base set $E$ and a set of forbidden subsets $F \subseteq P(E)$.
The forbidden subsets have the property that any superset of a
forbidden subset is also forbidden.
We let $x=\langle x_1, \ldots, x_{|E|} \rangle$ denote the request sequence.
There is a bijective function, $f$, mapping the $x_i$'s to $E$.
For a set $S=\{x_i,x_j,x_h,\ldots\}$, we let $f(S)$ denote $\{f(x_i),f(x_j),f(x_h),\ldots\}$.
This function $f$ is not known to the algorithm.
In request $i$, the algorithm receives request $x_i$.
The request contains a list of all minimal subsets $A \subseteq \{x_1, \ldots, x_{i-1}\}$
such that $f(A) \cup \{f(x_i)\} \in F$ (note that this list may be empty).
The algorithm must reject or accept $x_i$.
The produced solution is said to be feasible if it does not contain any subsets from $F$. 
The score of a feasible solution is the number of accepted elements.
The score of an infeasible solution is $-\infty$.
Note that if all minimal sets in $F$ have size two, this is equivalent to \IS.

In \MOS, an isolated element is an element from $E$ which is not in any sets of $F$.
Note that such an element can be added to any solution.
We let $s(E,F)$ denote a smallest $S \subseteq E$ such that adding any element to $S$
results in a set which contains a forbidden subset.

The greedy algorithm, \GMOS, is the algorithm which always accepts a request if the resulting solution is feasible.

For an algorithm, \alg, we let $\alg(E,F)$ be the smallest number such that there
exist an ordering of $E$ which causes \alg to accept at most $\alg(E,F)$ elements (using $F$ as forbidden subsets).
We let $\oms(E,F)$ be the largest number such that there exists an algorithm with $\alg(E,F)=\oms(E,F)$.
\begin{theorem}\label{mos}
Let $(E,F)$ be a \MOS instance, and let $E'$ be $E$ with the
isolated elements removed. Let $k$ denote the number of isolated elements.
If $k+|s(E',F)| \geq \oms(E',F)$, then
\[
	\GMOS(E,F)=\oms(E,F).
\]
\end{theorem}
\begin{proof}
This proof is similar to that of Theorem~\ref {isgod}.
First note that \GMOS accepts the $k$ isolated elements and at least $|s(E',F)|$ elements from $E'$.
For any algorithm the adversary can start by requesting
elements, each with an empty list of forbidden sets it is already contained
in. It continues until
the algorithm has either accepted $|s(E',F)|$ elements or rejected $k$.
The key argument is that the algorithm cannot distinguish between these
initial elements. 

If the algorithm accepts at least $|s(E',F)|$ elements,
the adversary can decide that they were exactly those in $s(E',F)$, which \GMOS  also accepts. In this case, the algorithm cannot
accept more than $k+|s(E',F)|$ elements in total. 

If the algorithm rejects $k$ elements, we need a result similar to that of Lemma~\ref{conservative}.
For \MOS, a pointless request is one, which reveals a forbidden set which contains only elements that have been accepted.
Accepting a pointless request would result in an infeasible solution. The same argument as in the proof for Lemma~\ref{conservative}
shows that an adversary loses no power by being conservative. Thus, when $k$ elements have been rejected (and up to $|s(E',F)|-1$ have been accepted),
the adversary has a strategy for the remaining elements which ensures that the algorithm accepts at most $\oms(E',F)\leq k+|s(E',F)| $ elements. \qed
\end{proof}

This problem is quite flexible. As we have mentioned, it can model independent set, but it could also model matroid intersection problems such as bipartite
matching (though, even with more than two matroids). In this case,  
the forbidden sets, $F$, are the dependent sets in the union of the matroids. 

\section{Implications for Worst Case Performance Measures}
\label{measures}
Do the results from the previous section mean that \GI is a good algorithm for \IS if the input graph is known to be a Freckle Graph? 
The answer to this depends on how the performance of online algorithms is measured. In general, the answer is yes
if a measure that only considers the worst case is used.

The most commonly used performance measure for online algorithms is \emph{competitive analysis} \cite{Sleator}.
For maximization problems, an algorithm, \alg, is said to be $c$-competitive if there exists a constant, $b$,
such that for any input sequence, $I$, $\opt(I) \leq c\alg(I)+b$ where $\opt(I)$ is the score of the optimal offline algorithm.
For minimization problems, we require that $\alg(I) \leq c\opt(I)+b$. 
The competitive ratio of \alg is $\inf \; \{c : \alg \text{ is $c$-competitive} \}$. (Note that these ratios are always at least $1$.)
For \emph{strict competitive analysis}, the definition is the same, except there is no additive constant.

Another measure is \emph{on-line competitive analysis} \cite{GKL97},
which was introduced for online graph coloring.
The definition is the same as for competitive analysis except that $\opt(I)$ is replaced by $\opton(I)$, which is the score of the
best online algorithm that knows the requests in $I$ but not their ordering. For graph problems, this means that the vertex-arrival
model is used, as in this paper. The algorithm is allowed to know the final graph.

\begin{corollary}
For \IS on Freckle Graphs, no algorithm has a smaller competitive ratio, strict competitive ratio, or on-line competitive ratio than \GI.
\end{corollary}
\begin{proof}
Let \alg be a $c$-competitive algorithm for some $c$. Theorem~\ref{isgod} implies that \GI is also $c$-competitive. This argument also holds for the strict competitive ratio and the
on-line competitive ratio. \qed
\end{proof}

\begin{corollary}
For \VC on Freckle Graphs, no algorithm has a smaller competitive ratio, strict competitive ratio, or on-line competitive ratio than \GV.
\end{corollary}

\begin{corollary}
For \DS on the class of graphs with at least one isolated vertex, no algorithm has a smaller competitive ratio, strict competitive ratio, or on-line competitive ratio than \GD.
\end{corollary}

Similar results hold for relative worst order analysis~\cite{BF07}. According to
relative worst order analysis, for minimization
problems in this graph model, one algorithm, $A$, is at least as good
as another algorithm, $B$, on a graph class, if for all graphs $G$ in
the class, $A(G) \leq B(G)$. The inequality is reversed for
maximization problems. 
It follows from the definitions
that if an algorithm is optimal with
respect to on-line competitive analysis, it is also optimal
with respect to relative worst-order analysis.
This was observed in~\cite{BFM}.
Thus, the above results show that the three greedy algorithms in the corollaries above
are also optimal on Freckle Graphs, under relative worst order analysis.

\section{A Subclass of Freckle Graphs Where Greedy Is Not Optimal (Under Some Non-Worst Case Measures)} \label{bijective}
Although these greedy algorithms are optimal with respect to
some worst case measures, this does not mean that these greedy algorithms
are always the best choice for \emph{all} Freckle Graphs.
There is a subclass of Freckle Graphs where another algorithm is objectively better than \GI,
and bijective analysis and average analysis~\cite{Angelopoulos07} reflect this.

\begin{theorem} \label{agi}
There exists an infinite class of Freckle Graphs $\family=\{G_n \; | n \geq 2\}$ and 
an algorithm \AGI such that for all $n \geq 2$ the following holds:
\begin{align*}
 \forall \phi& \; \AGI(\phi(G_n)) \geq \GI(\phi(G_n))\\
 \exists \phi& \; \AGI(\phi(G_n)) > \GI(\phi(G_n))
\end{align*}
\end{theorem}

\begin{proof}
Consider the graph $G_n=(V_n,E_n)$, where
\begin{align*}
V_n=&\{ x_1,x_2,\ldots,x_n,y_1,y_2,\ldots,y_n,z,u_1,u_2,\ldots,u_n\} \\
E_n=&\{ (x_i,y_i),(y_i,z),(z,u_i) \mid 1 \leq i \leq n\}.
\end{align*}
Figure~\ref{G4} shows the graph $G_4$.

\begin{figure}[h]
\centering
\begin{tikzpicture}[-,>=stealth',auto,node distance=1cm,
  thick,
  main node/.style={circle,fill=blue!20,minimum size=0.75cm,draw,font=\sffamily\Large\bfseries},
  uncolored/.style={circle,fill=white!20,minimum size=0.75cm,draw,font=\sffamily\Large\bfseries},
  smallnode/.style={circle,fill=blue!20,minimum size=0.75cm,draw,font=\sffamily\Large\bfseries}]

  \node(1)[main node] {$x_1$};
  \node(2)[main node] [right = 1 cm of 1] {$x_2$};
  \node(3)[main node] [right = 1 cm of 2] {$x_3$};
  \node(4)[main node] [right = 1 cm of 3] {$x_4$};
  
  \node(5)[main node] [below = 1 cm of 1] {$y_1$};
  \node(6)[main node] [right = 1 cm of 5] {$y_2$};
  \node(7)[main node] [right = 1 cm of 6] {$y_3$};
  \node(8)[main node] [right = 1 cm of 7] {$y_4$};
  \node(10) [right = 0.5 cm of 6] {};
  
  \node(9)[main node] [below = 1 cm of 10] {$z$};
  
  \node(11)[main node] [below = 2 cm of 5] {$u_1$};
  \node(12)[main node] [below = 2 cm of 6] {$u_2$};
  \node(13)[main node] [below = 2 cm of 7] {$u_3$};
  \node(14)[main node] [below = 2 cm of 8] {$u_4$};

  \path[every node/.style={font=\sffamily\small}]
	(1) edge (5)
	(2) edge (6)
	(3) edge (7)
	(4) edge (8)
	
	(5) edge (9)
	(6) edge (9)
	(7) edge (9)
	(8) edge (9)
	
	(9) edge (11)
	(9) edge (12)
	(9) edge (13)
	(9) edge (14)
	;
\end{tikzpicture}
\caption{The graph $G_4$.} \label{G4}
\end{figure}
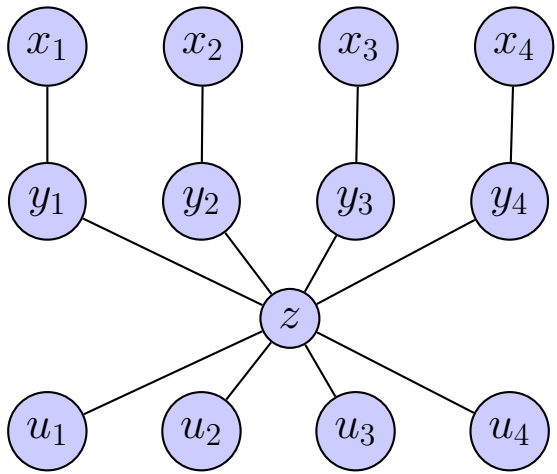

We start by showing that $G_n$ is a Freckle Graph.
The smallest maximal independent set has size $n+1$.
We want to show that $\ois(G)=n+1$, that is no algorithm can get an independent set of size more than $n+1$ in the worst case.
We consider an arbitrary algorithm, \alg, and
the situation where the adversary starts by presenting $n$ isolated vertices.
If \alg rejects all of these, the adversary can decide that it was $u_1, \ldots, u_n$.
In the remaining graph, it is not possible to accept more than $n+1$ vertices.
Otherwise, \alg accepts $i>0$ of the $n+1$ isolated vertices.
The adversary can decide that one was $z$ and that the remaining were $x_1, \ldots x_{i-1}$.
Since \alg accepted $z$, it can never accept  any of the vertices 
$y_1, \ldots, y_n$ or $u_1, \ldots, u_n$.
Thus, it can at most accept $n+1$ vertices.
This shows that $G_n$ is a Freckle Graph.

The algorithm, \AGI, is identical to \GI, except that
it rejects a vertex if it already has two neighbors when it is presented.
Consider any ordering of the vertices of $G$ where \GI and \AGI do not accept the same independent set.
There must exist a first vertex, $w$, which is accepted by one of the algorithms and rejected by the other.
By definition of the algorithms, it must be the case that $w$ is rejected by \AGI and accepted by \GI.
It must hold that $w$ has two neighbors, which have not been accepted by either algorithm.
This can only happen if $w=z$ and the two neighbors are $y_i$ and $y_j$ where $x_i$ and $x_j$ have already
been presented and accepted by both algorithms and no $u_k$ have been presented yet. In this case, $z$ is accepted
by \GI and rejected by \AGI. However, $u_1, \ldots, u_n$ are accepted by \AGI and rejected by \GI.
Since $n \geq 2$ and since both \GI and \AGI accept exactly one of
$x_i$ and $y_i$, $1\leq i\leq n$,
we get that on every ordering, $\phi$, where \GI and \AGI accept a different independent set, $\AGI(\phi(G)) > \GI(\phi(G))$.
Such an ordering always exists (the ordering $x_1,\ldots,x_n,y_1,\ldots,y_n,z,u_1,\ldots,u_n$ achieves this). \qed
\end{proof}

Competitive analysis, on-line competitive analysis, and relative worst order ratio do not identify \AGI as a better algorithm than \GI on the class of graphs \family defined
in the proof of Theorem~\ref{agi}. There are, however, other measures which do this. Bijective analysis and average analysis~\cite{Angelopoulos07} are such  measures. Let $I_n$ be the
set of all input sequences of length $3n+1$. Since we are considering the rather restricted graph class \family, $I_n$ denotes all orderings of the vertices in $G_n$
(since these are the only inputs of length $3n+1$). For an algorithm $A$ to be considered better than another algorithm $B$ for a maximization problem, it must hold for sufficiently large $n$
that there exists a bijection $f: \; I_n \rightarrow I_n$ such that the following holds:

\begin{align*}
\forall & I \in I_n \; A(I) \geq B(f(I)) \\
\exists & I \in I_n \; A(I) > B(f(I)) \\
\end{align*}

\begin{theorem}
\AGI is better than \GI on the class \family according to bijective analysis.
\end{theorem}
\begin{proof}
We let the bijection $f$ be the identity and the result follows from Theorem~\ref{agi}. \qed
\end{proof}
Average analysis is defined such that if one algorithm is better
than another according to bijective analysis, it is also better
according to average analysis. Thus, \AGI is better than \GI
on the class \family according to average analysis.

Note that \AGI is not an optimal algorithm for all Freckle Graphs.
The class of graphs, $K_{n,n}$, for $n\geq 2$,
 consisting of complete bipartite graphs with
$n$ vertices in each side of the partition, is a class where
\AGI can behave very poorly. Note that on these graphs, \GI
is optimal and always finds an independent set of size $n$,
which is optimal, so these graphs are Freckle Graphs, even
though they have no isolated vertices. If the first request
to \AGI is a vertex from one side of the partition and the
next two are from the other side of the partition, \AGI only
accepts one vertex, not $n$.

\section{Complexity of Determining the Online Independence Number, Vertex Cover Number, and Domination Number}
\label{hardness}
Given a graph, $G$, it is easy to check if it has an isolated vertex and apply Theorem~\ref{dsgod}.
However, Theorem~\ref{isgod} and Corollary~\ref{vcgod} might not be as 
easy to apply, because it is not obvious how one can check if
a graph is a Freckle Graph ($k+|\lille| \geq \ois(G')$). 
In some cases, this is easy.
 For example, any graph where at least half the vertices are 
isolated is a Freckle Graph.
We leave the hardness of recognizing Freckle Graphs as an open
problem, but
we show a hardness result for deciding if $\ois(G) \leq q$.

\begin{theorem} \label{ishard}
Given $q \in \mathbb{N}$ and a graph, $G$, deciding if $\ois(G) \leq q$ is \np-hard.
\end{theorem}
\begin{proof}
Note that it is \np-complete to determine
if the minimum maximal independent set of a graph, $G=(V,E)$, 
has size at most $L$, 
for an integer $L$~\cite{GJ79}. To reduce from this problem,
we create $\tilde{G}=(\tilde{V},E)$ which is the same as $G$, but has $|V|$ extra
isolated vertices, and a bound $\tilde{L}=L+|V|$. $\tilde{G}$ is a Freckle Graph,
since $|V| \geq \ois(G)$. By Corollary~\ref{freckle;greedy}, 
$\GI(\tilde{G})=\ois(\tilde{G})$. Since $\GI(\tilde{G}) = |s(G)|+|V|$,
the original graph, $G$, has a minimum maximal independent
set of size $L$, if and only if $\tilde{G}$ has online independence
number at most $\tilde{L}$.
 \qed
\end{proof}

The hardness of computing the online independence number implies
the hardness of computing the online vertex cover number.

\begin{corollary}
Given $q \in \mathbb{N}$ and a graph, $G$, deciding if $\ovc(G) \geq q$ is \np-hard.
\end{corollary}
\begin{proof}
This follows from Observation~\ref{isvcsum} and Theorem~\ref{ishard}. \qed
\end{proof}

\begin{theorem}
Given $q \in \mathbb{N}$ and a graph, $G$, deciding if $\ods(G) \geq q$ is \np-hard.
\end{theorem}
\begin{proof}
We make a reduction from \ISoff. In \ISoff, a graph, $G$ and an $L \in \mathbb{N}$ is given.
It is a yes-instance if and only if there exists an independent set of size at least $L$.
We reduce instances of \ISoff, ($G$,$L$), to instances of \DS, ($\tilde{G}$,$\tilde{L}$), such that
there exists an independent set in $G$ of size at least $L$ if and only if $\ods(\tilde{G}) \geq \tilde{L}$.
The reduction is very simple. We let $\tilde{G}$ be the graph which consists of $G$ with one additional isolated vertex.
We set $\tilde{L}=L+1$.
Assume first that any independent set in $G$ has size at most $L-1$. This means that any independent set in $\tilde{G}$ has
size at most $L$. Since \GD produces an independent set, it will accept at most $L<\tilde{L}$ vertices.
Assume now that there is an independent set of size at least $L$ in $G$. Then, there exists an independent set of size at least $L+1$
in $\tilde{G}$. If these vertices are presented first, \GD will accept them. From Theorem~\ref{dsgod}, we get that no algorithm for
\DS can do better (since $\tilde{G}$ has an isolated vertex), which means that $\ods(\tilde{G}) \geq \tilde{L}$. \qed
\end{proof}

\begin{theorem} \label{ispspace}
Given $q \in \mathbb{N}$ and a graph, $G$, the problem of deciding if $\ois(G) \leq q$ is in \pspace.
\end{theorem}
\begin{proof}
Let $q \in \mathbb{N}$ and a graph, $G=(V,E)$, be given. We sketch an algorithm that uses only polynomial space which decides if $\ois(G) \leq q$.
We view the problem as a game between the adversary and the algorithm where the algorithm wins if it gets an independent set of size at least $q+1$.
A move for the adversary is revealing a vertex along with edges to a subset of the previous vertices such that the resulting graph
is an induced subgraph of $G$.
 These are possible to enumerate since induced subgraph can be solved in polynomial space. A move by the algorithm is accepting or rejecting that vertex.
We make two observations: The game has only polynomial length (each game has length $2|V|$),
 and it is always possible in polynomial space to enumerate the possible moves from a game state.
Thus, an algorithm can traverse the game tree using depth first search and recursively compute for each game state if the adversary or the algorithm has a winning strategy. \qed
\end{proof}
Similar proofs can be used to show that the problems of deciding if $\ovc(G) \geq q$ and $\ods(G) \geq q$ are in \pspace as well.
It remains open whether these problems are \np-complete, \pspace-complete, or neither.
In \cite{Kudahl15} it was shown that determining the online chromatic number is \pspace-complete if the graph is pre-colored
and extended in \cite{BV16} to hold even if the graph is not pre-colored.

\section{Concluding Remarks}
A strange difference between online and offline algorithms
is observed: Adding isolated vertices to a graph can change
an algorithm from not being optimal to being optimal (according
to many measures). This holds for \IS, \VC, and \DS. It is also
shown that adding isolated elements can make the natural greedy
algorithm optimal for  \M and \MOS 
(which includes \IM as a special case), but
not for all problems in the class AOC. 

It is even more surprising
that this difference occurs for vertex
cover than for independent set, since in the offline case,
adding isolated vertices to a graph  can improve the approximation
ratio in the case of the independent set problem. It is hard to
see how adding isolated vertices to a graph could in any way
help an offline algorithm for vertex cover.

We have shown that for Freckle Graphs, the greedy algorithm is optimal for \IS,
but what about the converse?
If a graph is not Freckle, is it the case that the greedy algorithm is not optimal?
Let $G$ be a graph, that is not a Freckle Graph.
By definition, we have that $\ois(G') > |\lille|+k=\GI(G)$.
To show that the greedy algorithm is not optimal, we
would have to show that $\ois(G) > \GI(G)$.
To show this, it would suffice to show that $\ois(G) \geq \ois(G')$.
That is, the online independence number can never decrease when isolated vertices are added to a graph.
We leave this as an open question.

Note that $\GI=\GD$. This means that for Freckle Graphs with at least one isolated vertex, $\GI$ is an algorithm
which solves both online independent set (a maximization problem) and
online dominating set (a minimization problem) online optimally. This is
quite unusual, since the independent sets and dominating sets it will find
in the worst case can be quite different for the same graph.


As mentioned earlier, the NP-hardness results presented here do not answer the question as to how hard it
is to recognize Freckle Graphs. This is left as an open problem. 

We have shown it to be \np-hard to decide if $\ois(G) \leq q$, $\ovc(G) \geq q$, and $\ods(G) \geq q$, but there is nothing to suggest that these problems are contained in \np.
They are in \pspace, but it is left as an open problem if they are \np-complete, \pspace-complete or somewhere in between.

\section*{Acknowledgments}
The authors would like to thank Lene Monrad Favrholdt for interesting and 
helpful discussions. This research was supported by grants from the 
Villum Foundation (VKR023219), and the Danish Council for Independent Research, Natural Sciences (DFF–1323-00247). The second author was also supported
by a travel stipend from the Stibo-Foundation.

\bibliographystyle{splncs03}
\bibliography{refs}

\end{document}